\documentclass[preprint,authoryear,3p,11pt]{elsarticle}

\usepackage{amssymb}
\usepackage{amsthm}
\usepackage{amsmath}

\usepackage{mathrsfs}
\usepackage{float}
\usepackage{eufrak}
\usepackage{mathtools}
\usepackage{graphics}
\usepackage{amsfonts}
\usepackage{mathrsfs}
\usepackage{amscd}
\usepackage{color}
\usepackage{url}
\usepackage[plainpages=false,pdfpagelabels=true,colorlinks=true,citecolor=blue,hypertexnames=false]{hyperref}
\usepackage[ruled,vlined]{algorithm2e}
\usepackage{multirow,booktabs}

\journal{Multidimensional Systems and Signal Processing}

\newtheorem{theorem}{Theorem}[section]

\newtheorem{corollary}[theorem]{Corollary}
\newtheorem{lemma}[theorem]{Lemma}
\newtheorem{remark}[theorem]{Remark}
\newtheorem{definition}[theorem]{Definition}
\newtheorem{example}[theorem]{Example}
\newtheorem{question}[theorem]{Question}
\numberwithin{equation}{section}

\newcommand{\gr}{Gr\"{o}bner }
\newcommand{\z}{{\bf z}}

\begin{document}

\begin{frontmatter}



\title{Factorizations for a Class of Multivariate Polynomial Matrices}


\author[klmm,ucas]{Dong Lu}
\ead{donglu@amss.ac.cn}
%\cortext[cor]{Corresponding author}

\author[klmm,ucas]{Dingkang Wang}
\ead{dwang@mmrc.iss.ac.cn}

\author[klmm,ucas]{Fanghui Xiao}
\ead{xiaofanghui@amss.ac.cn}

\address[klmm]{KLMM, Academy of Mathematics and Systems Science, Chinese Academy of Sciences, Beijing 100190, China}

\address[ucas]{School of Mathematical Sciences, University of Chinese Academy of Sciences, Beijing 100049, China}

\begin{abstract}
 Following the works by Lin et al. (Circuits Syst. Signal Process. 20(6): 601-618, 2001) and Liu et al. (Circuits Syst. Signal Process. 30(3): 553-566, 2011), we investigate how to factorize a class of multivariate polynomial matrices. The main theorem in this paper shows that an $l\times m$ polynomial matrix admits a factorization with respect to a polynomial if the polynomial and all the $(l-1)\times (l-1)$ reduced minors of the matrix generate the unit ideal. This result is a further generalization of previous works, and based on this, we give an algorithm which can be used to factorize more polonomial matrices. In addition, an illustrate example is given to show that our main theorem is non-trivial and valuable.
 \end{abstract}


\end{frontmatter}

\section{Introduction}\label{intro}

 The study of factorizations for multivariate polynomial matrices began with the development of multidimensional system theory in the late 1970s \citep{Youla1979Notes}, and the problem of matrix factorizations was considered to be one of the basic problems of this subject. Since then, great progress has been made on multivariate polynomial matrix factorizations.

 \cite{Bose1982} introduced some basic concepts of multivariate polynomial matrices and the problem of matrix factorizations. After that, \cite{Bose2003} presented factorization algorithms of bivariate polynomial matrices, and introduced the latest research trends of matrix factorizations with three or more variables. The factorization problem for bivariate polynomial matrices has been completely solved in \citep{Guiver1982Polynomial,Liu2013New,Morf1977New}, but for the cases of more than two variables is still open.

 \cite{Charoenlarpnopparut1999Multidimensional} used \gr bases of modules to compute zero prime matrix factorizations of multivariate polynomial matrices. For some polynomial matrices with special properties, \cite{Lin1999Notes,Lin2001Further} proposed some methods to compute zero prime matrix factorizations of matrices. Meanwhile, \cite{Lin2001A} presented the Lin-Bose's conjecture: a matrix admits a zero prime matrix factorization if its all maximal reduced minors generate the unit ideal. This conjecture was proved in \citep{Liu2014The,Pommaret2001Solving,Wang2004On}, so the problem of zero prime matrix factorizations have been completely solved. Subsequently, \cite{Mingsheng2005On} put forward an algorithm based on module theory to solve the problem of minor prime matrix factorizations. \cite{Guan2018,Guan2019} studied the problem of factor prime matrix factorizations under the condition that matrices are not of full rank, and they generalized the main results in \citep{Mingsheng2005On}. So far, some achievements in \citep{Liu2010Notes,Liu2015Further,Mingsheng2007On} have been made on the problem of factor prime matrix factorizations. Although the problems of zero prime matrix factorizations and minor prime matrix factorizations have been completely solved, the problem of factor prime matrix factorizations remains to be studied.

 Let $k[\z]$ and $k[\z_2]$ be the ring of polynomials in variables $z_1,z_2,\ldots,z_n$ and $z_2,\ldots,z_n$ with coefficients in an algebraically closed field $k$, respectively. Let $\mathbf{F}$ be an $l\times m$ polynomial matrix with entries in $k[\z]$ and $l\leq m$, $d_l(\mathbf{F})$ be the greatest common divisor of all the $l\times l$ minors of $\mathbf{F}$, and $d = z_1 - f(\z_2)$ be a divisor of $d_l(\mathbf{F})$, where $f(\z_2)\in k[\z_2]$. \cite{Lin2001Factorizations} proved that $\mathbf{F}$ admits a matrix factorization with respect to $d$ if for each $(\z_2)\in k^{1\times (n-1)}$ the rank of $\mathbf{F}(f(\z_2),\z_2)$ is $(l-1)$. Moreover, they proposed a constructive algorithm to factorize this class of multivariate polynomial matrices. \cite{Liu2011On} focused on the relationship between $d$ and all the $(l-1)\times (l-1)$ minors of $\mathbf{F}$, and showed that $\mathbf{F}$ admits a matrix factorization with respect to $d$ if $d$ and all the $(l-1)\times (l-1)$ minors of $\mathbf{F}$ generate $k[\z]$. They proved that their main theorem is a generalization of the result in \citep{Lin2001Factorizations}. However, we find that there are still many of multivariate polynomial matrices that can be factorized with respect to $d$ without satisfying the main theorem in \citep{Liu2011On}. This implies that it would be significant to generalize the theorems and algorithms in \citep{Lin2001Factorizations,Liu2011On}.

 In this paper, we still study the condition under which $\mathbf{F}$ admits a matrix factorization with respect to $d$. We focus on the relationship between $d$ and all the $(l-1)\times (l-1)$ reduced minors of $\mathbf{F}$, and prove that $\mathbf{F}$ admits a matrix factorization with respect to $d$ if $d$ and all the $(l-1)\times (l-1)$ reduced minors of $\mathbf{F}$ generate the unit ideal. Compared with the main theorems in \citep{Lin2001Factorizations,Liu2011On}, our main theorem has a wider range of applications. Combining our main theorem and the constructive algorithm in \citep{Lin2001Factorizations}, we obtain the matrix factorization algorithm.

 This paper is organized as follows. In Section \ref{sec2}, we outline some knowledge about multivariate polynomial matrix factorizations and propose a problem that we shall consider. Main theorem and some generalizations are presented in Section \ref{sec3} to help us summarize which types of polynomial matrices can be factorized. The matrix factorization algorithm is given in Section \ref{sec4}, and an example is given to illustrate the calculation process of the algorithm. Further remarks are provided in Section \ref{sec_conclusions}.

\section{Preliminaries and Problem}\label{sec2}

 In the following, we denote by $k$ an algebraically closed field, $\z$ the $n$ variables $z_1,z_2,\ldots,z_n$, $\z_2$ the $(n-1)$ variables $z_2,\ldots,z_n$, where $n\geq 3$. Let $k[\z]$ and $k[\z_2]$ be the ring of polynomials in variables $\z$ and $\z_2$ with coefficients in $k$, respectively, $k(\z)$ be the fraction field of $k[\z]$, and $k[\z]^{l\times m}$ be the set of $l\times m$ matrices with entries in $k[\z]$. Without loss of generality, we assume that $l\leq m$, and for convenience we use uppercase bold letters to denote polynomial matrices.

 Throughout this paper, the argument $(\z)$ is omitted whenever its omission does not cause confusion. For any given $\mathbf{F}\in k[\z]^{l\times m}$ and $f(\z_2)\in k[\z_2]$, $\mathbf{F}^{{\rm T}}$ represents the transposed matrix of $\mathbf{F}$, and $\mathbf{F}(f(\z_2),\z_2)$ denotes an $l\times m$ polynomial matrix in $k[\z_2]^{l\times m}$, which is formed by transforming $z_1$ in $\mathbf{F}$ into $f(\z_2)$. If $l = m$, we denote by ${\rm det}(\mathbf{F})$ the determinant of $\mathbf{F}$, and if $\mathbf{F}$ is of full rank, we use $\mathbf{F}^{-1}\in k(\z)^{l\times l}$ to stand for the inverse matrix of $\mathbf{F}$. Assume that $f_1,\ldots,f_s\in k[\z]$, we use $\langle f_1,\ldots,f_s \rangle$ to denote the ideal generated by $f_1,\ldots,f_s$ in $k[\z]$. Let $f,g \in k[\z]$, then $f \mid g$ means that $f$ is a divisor of $g$. In addition, ``w.r.t." and ``GCD'' stand for ``with respect to" and ``greatest common divisor", respectively.

\subsection{Previous Works}\label{sec2-1}

 We first introduce two basic concepts in matrix theory.

\begin{definition}\label{def_minor}
 Let $\mathbf{F} \in k[\z]^{l\times m}$, and given $2r$ positive integers arbitrarily such that $1\leq i_1 < \cdots < i_r \leq l$ and  $1\leq j_1 < \cdots < j_r \leq m$. Let $\mathbf{F}\begin{pmatrix}\begin{smallmatrix}
 i_1\cdots i_r \\j_1\cdots j_r\end{smallmatrix}\end{pmatrix}$ denotes an $r\times r$ matrix consisting of the $i_1,\ldots,i_r$ rows and $j_1,\ldots,j_r$ columns of $\mathbf{F}$, then ${\rm det}\Bigl(\mathbf{F}\begin{pmatrix}\begin{smallmatrix}i_1\cdots i_r \\j_1\cdots j_r\end{smallmatrix}\end{pmatrix}\Bigr)$ is called an $r\times r$ minor of $\mathbf{F}$.
\end{definition}

\begin{definition}\label{def_rank}
 Let $\mathbf{F} \in k[\z]^{l\times m}$, the rank of $\mathbf{F}$ is $r$ $(1\leq r \leq l)$ if there exists a nonzero $r\times r$ minor of $\mathbf{F}$, and all the $i\times i$ $ (i> r)$ minors of $\mathbf{F}$ vanish identically. For convenience, we denote the rank of $\mathbf{F}$ by ${\rm rank}(\mathbf{F})$.
\end{definition}

 The following lemma is a generalization of Binet-Cauchy formula in \citep{Strang2010Linear}.

\begin{lemma}\label{binet-cauchy}
 Let $\mathbf{F}=\mathbf{G}_1\mathbf{F}_1\in k[\z]^{l\times m}$, where $\mathbf{G}_1\in k[\z]^{l\times l}$ and $\mathbf{F}_1 \in k[\z]^{l\times m}$. Then an $r\times r$ $(r\leq l)$ minor of $\mathbf{F}$ is
 \begin{equation}\label{binet-cauchy-equation-1}
 {\rm det}\Bigl(\mathbf{F}\begin{pmatrix}\begin{smallmatrix}
  i_1\cdots i_r \\ j_1\cdots j_r \end{smallmatrix}\end{pmatrix}\Bigr) \\  =  \sum_{1\leq s_1<\cdots<s_r\leq l}{\rm det}\Bigl(\mathbf{G}_1\begin{pmatrix}\begin{smallmatrix}
  i_1\cdots i_r \\ s_1\cdots s_r
  \end{smallmatrix}\end{pmatrix}\Bigr)\cdot {\rm det}\Bigl(\mathbf{F}_1\begin{pmatrix}\begin{smallmatrix}
  s_1\cdots s_r \\ j_1\cdots j_r\end{smallmatrix}\end{pmatrix}\Bigr).
 \end{equation}
 In particular, when $r = l$, we have
 \begin{equation}\label{binet-cauchy-equation-2}
  {\rm det}\Bigl(\mathbf{F}\begin{pmatrix}\begin{smallmatrix}
    1~\cdots ~l \\ j_1\cdots j_l
    \end{smallmatrix}\end{pmatrix}\Bigr) = {\rm det}(\mathbf{G}_1)\cdot {\rm det}\Bigl(\mathbf{F}_1\begin{pmatrix}\begin{smallmatrix}
    1~\cdots ~l \\j_1\cdots j_l\end{smallmatrix}\end{pmatrix}\Bigr).
 \end{equation}
\end{lemma}

 To make the description simpler, we use the notations and concepts in the paper \citep{Lin1988On}.

\begin{definition}
 For any given $\mathbf{F}\in k[\z]^{l\times m}$,
 \begin{enumerate}
    \item let $d_l(\mathbf{F})$ and $d_{l-1}(\mathbf{F})$ be the GCD of all the $l\times l$ minors and all the $(l-1)\times (l-1)$ minors of $\mathbf{F}$, respectively;

    \item let $a_1,\ldots,a_\eta \in k[\z]$ be all the $l\times l$ minors of $\mathbf{F}$, where $\eta = \binom m{l}$, and extracting $d_l(\mathbf{F})$ from $a_1,\ldots,a_\eta$ yields
        $$a_i = d_l(\mathbf{F})b_i,~~i=1,\ldots,\eta,$$
        then $b_1,\ldots,b_\eta$ are called the $l\times l$ reduced minors of $\mathbf{F}$;

    \item let $c_1,\ldots,c_\gamma \in k[\z]$ be all the $(l-1)\times (l-1)$ minors of $\mathbf{F}$, where $\gamma = \binom {l}{l-1} \cdot \binom {m}{l-1}$, and extracting $d_{l-1}(\mathbf{F})$ from $c_1,\ldots,c_\gamma$ yields
        $$c_i = d_{l-1}(\mathbf{F})h_i,~~i=1,\ldots,\gamma,$$
        then $h_1,\ldots,h_\gamma$ are called the $(l-1)\times (l-1)$ reduced minors of $\mathbf{F}$.
 \end{enumerate}
\end{definition}

 Now, we introduce two important lemmas in matrix theory.

\begin{lemma}[\cite{Lin1993On,Lin1999On}]\label{RM_relation-1}
 Let $\mathbf{F}_1=[\mathbf{F}_{11},\mathbf{F}_{12}] \in k[\z]^{l\times (m+l)}$ be of full row rank and $\mathbf{F}_2=[\mathbf{F}_{21}^{{\rm T}},-\mathbf{F}_{22}^{{\rm T}}]^{{\rm T}}$ $\in k[\z]^{(m+l)\times m}$ be of full column rank, where $\mathbf{F}_{11},\mathbf{F}_{22}\in k[\z]^{l\times m}$, $\mathbf{F}_{12} \in k[\z]^{l\times l}$ and $\mathbf{F}_{21} \in k[\z]^{m\times m}$. If $\mathbf{F}_1\mathbf{F}_2=\mathbf{0}_{l\times m}$, then ${\rm det}(\mathbf{F}_{12})\neq 0$ if and only if ${\rm det}(\mathbf{F}_{21})\neq 0$.
\end{lemma}

\begin{lemma}[\cite{Lin1988On}]\label{RM_relation-2}
 Assume that $\mathbf{F}_{12}^{-1}\mathbf{F}_{11}=\mathbf{F}_{22}\mathbf{F}_{21}^{-1}$, where $\mathbf{F}_{11},\mathbf{F}_{22} \in k[\z]^{l\times m}$, $\mathbf{F}_{12}^{-1}\in k(\z)^{l\times l}$ and $\mathbf{F}_{21}^{-1}\in k(\z)^{m\times m}$. Let $\bar{p}_1,\ldots,\bar{p}_{\xi_1}$ be all the $l\times l$ reduced minors of $[\mathbf{F}_{11},\mathbf{F}_{12}]$, and $p_1,\ldots,p_{\xi_2}$ be all the $m\times m$ reduced minors of $[\mathbf{F}_{21}^{{\rm T}},-\mathbf{F}_{22}^{{\rm T}}]^{{\rm T}}$, where $\xi_1 = \binom {m+l}{l} = \xi_2 = \binom {m+l}{m}$. Then, $\bar{p}_i=\pm p_i$ for $i=1,\ldots,\xi_1$, and the sign depends on the index $i$.
\end{lemma}

 The general matrix factorization problem is now formulated as follows.

\begin{definition}\label{matrix_factorization}
 Let $\mathbf{F}\in k[\z]^{l\times m}$ and $d\in k[\z]$ be a divisor of $d_l(\mathbf{F})$. We say that $\mathbf{F}$ admits a matrix factorization w.r.t. $d$ if $\mathbf{F}$ can be factorized as
 $$\mathbf{F} = \mathbf{G}_1\mathbf{F}_1$$
 such that $\mathbf{G}_1\in k[\z]^{l\times l}$, $\mathbf{F}_1\in k[\z]^{l\times m}$, and ${\rm det}(\mathbf{G}_1) = d$.
\end{definition}

Next we recall the concept of zero left prime matrix from multidimensional systems theory.

\begin{definition}
 Let $\mathbf{F}\in k[\z]^{l\times m}$ be of full row rank. If all the $l\times l$ minors of $\mathbf{F}$ generate $k[\z]$, then $\mathbf{F}$ is said to be a zero left prime (ZLP) matrix.
\end{definition}

 In Definition \ref{matrix_factorization} if $\mathbf{F}_1$ is a ZLP matrix, then we say that $\mathbf{F}$ admits a ZLP matrix factorization. Let $I$ be an ideal generated by all the $l\times l$ minors of $\mathbf{F}$, then we can compute the reduced \gr basis $\mathcal{G}$ of $I$ w.r.t. a term order to check $I=k[\z]$. That is, if $\mathcal{G} = \{ 1\}$, then $I = k[\z]$. The definition of reduced \gr basis and how to compute a reduced \gr basis of an ideal can be found in \citep{buchberger1965algorithmus,Cox2007Ideals}.

 \cite{serre1955faisceaux} raised the question whether any finitely generated projective module over a polynomial ring is free. This question was solved positively and independently by \cite{Quillen1976Projective} and \cite{Suslin1976Projective}, and the result is called Quillen-Suslin theorem. For Quillen-Suslin theorem, there are two descriptions as follows.

\begin{lemma} \label{QS-1}
 If $\mathbf{w}\in k[\z]^{1 \times l}$ is a {\rm ZLP} vector, then the set $\mathbf{M}\subset k[\z]^{l\times 1}$ constructed by all solutions $\mathbf{q}\in k[\z]^{l\times 1}$ of $\mathbf{w}\mathbf{q}=0$ is free.
\end{lemma}

\begin{lemma} \label{QS-2}
 If $\mathbf{w}\in k[\z]^{1 \times l}$ is a {\rm ZLP} vector, then an unimodular matrix $\mathbf{U}\in k[\z]^{l\times l}$ can be constructed such that $\mathbf{w}$ is its first row.
\end{lemma}

 In Lemma \ref{QS-1}, $\mathbf{M}$ is called the syzygy module of $\mathbf{w}$. \cite{Fabianska2006Applications} gave an algorithm to compute free bases of free modules over polynomial rings, and the algorithm was implemented in QuillenSuslin package \citep{QS-program}. In Lemma \ref{QS-2}, $\mathbf{U}$ is an unimodular matrix if and only if ${\rm det}(\mathbf{U})$ is a nonzero constant in $k$. There are many methods to construct $\mathbf{U}$ such that $\mathbf{w}$ is its first row, we refer to \citep{Logar1992Algorithms,Lu2017,park1995,Youla1984The} for more details.

\subsection{Problem}\label{sec2-2}

 In order to raise the problem we are going to consider, let us first introduce the works in \citep{Lin2001Factorizations} and \citep{Liu2011On}.

\begin{lemma}[\cite{Lin2001Factorizations}]\label{theorem_Lin}
 Let $\mathbf{F}\in k[\z]^{l\times m}$, and $d = z_1 - f(\z_2)$ be a common divisor of $a_1,\ldots,a_\eta$, i.e., $a_i = de_i$ with $e_i\in k[\z]$ $(i=1,\ldots, \eta)$. If $\langle d,e_1,\ldots,e_\eta \rangle = k[\z]$, then ${\rm rank}(\mathbf{F}(f(\z_2),
 \z_2))=l-1$ for every $(\z_2)\in k^{1\times (n-1)}$ and $\mathbf{F}$ admits a matrix factorization w.r.t. $d$.
\end{lemma}

 \cite{Liu2011On} proved that ${\rm rank}(\mathbf{F}
 (f(\z_2),\z_2))=l-1$ for every $(\z_2)\in k^{1\times (n-1)}$ if and only if $\langle d,c_1,\ldots,c_\gamma \rangle = k[\z]$. Therefore, they generalized Lemma \ref{theorem_Lin} and obtained the following result.

\begin{lemma}[\cite{Liu2011On}]\label{theorem_Liu}
 Let $\mathbf{F}\in k[\z]^{l\times m}$, and $d = z_1 - f(\z_2)$ be a divisor of $d_l(\mathbf{F})$. If $\langle d,c_1,\ldots,c_\gamma \rangle = k[\z]$, then ${\rm rank}(\mathbf{F}(f(\z_2),
 \z_2))=l-1$ for every $(\z_2)\in k^{1\times (n-1)}$ and $\mathbf{F}$ admits a matrix factorization w.r.t. $d$.
\end{lemma}

 In the following, let $d = z_1 - f(\z_2)$ with $f(\z_2)\in k[\z_2]$. According to Lemma \ref{theorem_Lin} and Lemma \ref{theorem_Liu}, we construct two sets of multivariate polynomial matrices:
 \begin{equation*}
   \left\{
     \begin{array}{ll}
       \mathcal{S}_1 & := \{\mathbf{F}\in k[\z]^{l\times m} : d \mid d_l(\mathbf{F})  \text{ and } \langle d,e_1,\ldots,e_\eta \rangle = k[\z]\}, \\
       \mathcal{S}_2 & := \{\mathbf{F}\in k[\z]^{l\times m} : d \mid d_l(\mathbf{F})  \text{ and } \langle d,c_1,\ldots,c_\gamma \rangle = k[\z]\}.
     \end{array}
   \right.
 \end{equation*}
 Then, we have $\mathcal{S}_1 \subset \mathcal{S}_2$ and $\mathbf{F}\in \mathcal{S}_2$ admits a matrix factorization w.r.t. $d$. Example 1 in the Section 4 of \citep{Lin2001Factorizations} shows that $\mathcal{S}_1$ is not empty, and Example 4.1 in the Section 4 of \citep{Liu2011On} shows that $\mathcal{S}_1 \subsetneqq \mathcal{S}_2$.

 Lemma \ref{theorem_Liu} tell us that for any given $\mathbf{F}\in \mathcal{S}_2$, ${\rm rank}(\mathbf{F}(f(\z_2),
 \z_2))=l-1$. This implies that ${\rm GCD}(d,d_{l-1}(\mathbf{F}))=1$. Otherwise, it follows from $d$ is an irreducible polynomial that ${\rm GCD}(d,d_{l-1}(\mathbf{F}))$ $=d$, then $c_i(f(\z_2),\z_2)=0~
 (i=1,\ldots,\gamma)$ and ${\rm rank}(\mathbf{F}(f(\z_2),
 \z_2))<l-1$, which leads to a contradiction. Now, we construct a new set of multivariate polynomial matrices:
 $$\mathcal{S} :=\{\mathbf{F}\in k[\z]^{l\times m} : d \mid d_l(\mathbf{F}) \text{ and } {\rm GCD}(d,d_{l-1}(\mathbf{F}))=1\}.$$
 Then, $\emptyset \neq \mathcal{S}_1 \subsetneqq \mathcal{S}_2 \subset \mathcal{S}$. As we know, $d_{l-1}(\mathbf{F})$ is the GCD of $c_1,\ldots,c_\gamma$, then we have
 $$\langle d,c_1,\ldots,c_\gamma \rangle \subseteq \langle d,d_{l-1}(\mathbf{F}) \rangle \subseteq k[\z].$$
 Therefore, it follows that $\langle d,c_1,\ldots,c_\gamma \rangle \neq k[\z]$ if $\langle d,d_{l-1}(\mathbf{F}) \rangle \neq k[\z]$. Although ${\rm GCD}(d,d_{l-1}(\mathbf{F}))=1$ for $\mathbf{F}\in \mathcal{S}$, $d$ and $d_{l-1}(\mathbf{F})$ may have common zeros. Next, we give an example to show that there exits $\mathbf{F}\in \mathcal{S}\setminus \mathcal{S}_2$ such that $\mathbf{F}$ admits a matrix factorization w.r.t. $d$.

\begin{example}\label{example-1}
 Let
 \begin{small}
  \[\mathbf{F}=\begin{bmatrix}
  z_1z_2-z_1-z_2^2-z_2z_3  &  z_1z_3+z_1-z_2z_3-z_2-z_3^2-z_3  & \mathbf{F}[1,3] \\
  -z_1z_2-z_1z_3+z_2+z_3 & z_2+z_3  &  z_1z_2+z_1z_3  \\
  z_1   &       -z_1+z_2+z_3 &         -2z_1+z_2+z_3+1
  \end{bmatrix}, \]
 \end{small}
 where $\mathbf{F}[1,3] = -z_1z_2+z_1z_3+2z_1+z_2^2-z_2-z_3^2-2z_3-1$.

 It is easy to compute that $d_3(\mathbf{F})= (z_1-z_2)(z_2+z_3)^2$ and $d_2(\mathbf{F}) =z_2+z_3$. Let $d = z_1-z_2$, then $d \mid d_3(\mathbf{F})$ and ${\rm GCD}(d,d_2(\mathbf{F}))=1$. Hence, $\mathbf{F}\in \mathcal{S}$.

 $a_1= (z_1-z_2)(z_2+z_3)^2$ is the $3\times 3$ minor of $\mathbf{F}$, and extracting $d$ from $a_1$ yields $e_1=(z_2+z_3)^2$. It is easy to check that the reduced \gr basis of $\langle d, e_1 \rangle$ w.r.t. the lexicographic order is $\{z_1-z_2, (z_2+z_3)^2\}$, then $\mathbf{F}\notin \mathcal{S}_1$.

 Since the reduced \gr basis of $\langle d,d_2(\mathbf{F}) \rangle$ w.r.t. the lexicographic order is $\{z_1+z_3,z_2+z_3\}$, we have $\langle d, c_1,\ldots,c_9 \rangle \subseteq \langle d, d_2(\mathbf{F})\rangle \neq k[\z]$. Then, $\mathbf{F}\notin \mathcal{S}_2$.

 However, we can get a matrix factorization of $\mathbf{F}$ w.r.t. $d$:
 \begin{small}
  \[\mathbf{F} = \begin{bmatrix}
     d & 0 & -z_3 - 1 \\
     0 & 1 & 0 \\
     0 & 0 & 1
  \end{bmatrix}
  \begin{bmatrix}
     z_2 + z_3 & 0 & -z_2 - z_3 \\
     -z_1z_2-z_1z_3+z_2+z_3 & z_2+z_3  &  z_1z_2+z_1z_3  \\
     z_1   &       -z_1+z_2+z_3 &         -2z_1+z_2+z_3+1
  \end{bmatrix}.\]
 \end{small}
\end{example}

 In Example \ref{example-1}, we find that the reduced \gr basis of $\langle d, h_1,\ldots,h_9 \rangle$ w.r.t. the lexicographic order is $\{1\}$. In spire of it, we consider the following question.

\begin{question}\label{question-1}
 Let $\mathbf{F}\in \mathcal{S}$. If $\langle d,h_1,\ldots,h_\gamma \rangle = k[\z]$, does $\mathbf{F}$ have a matrix factorization w.r.t. $d$?
\end{question}

\section{Main Results} \label{sec3}

 Before giving the main theorem, we introduce two important lemmas.

\begin{lemma}[\cite{Lin2001Factorizations}] \label{lemma_zero}
 Let $g\in k[\z]$ and $f(\z_2)\in k[\z_2]$. If $g(f,z_2, \ldots,z_n)$ is a zero polynomial in $k[\z_2]$, then $(z_1 - f(\z_2))$ is a divisor of $g$.
\end{lemma}

 The following lemma is a generalization of Lemma 2 in \citep{Lin2001Factorizations}.

\begin{lemma} \label{lemma_new}
 Let $\mathbf{F}\in k[\z]^{l\times m}$ with ${\rm rank}(\mathbf{F})=l-1$. If $\langle h_1,\ldots,h_\gamma \rangle = k[\z]$, then there is a {\rm ZLP} vector $\mathbf{w}\in k[\z]^{1 \times l}$ such that $\mathbf{w}\mathbf{F} = \mathbf{0}_{1\times m}$.
\end{lemma}
\begin{proof}
 In view of ${\rm rank}(\mathbf{F})=l-1$, we could assume that the first $(l-1)$ row vectors $\mathbf{f}_1,\ldots,\mathbf{f}_{l-1}$ of $\mathbf{F}$ are $k[\z]$-linearly independent. This implies that $\mathbf{f}_1,\ldots,\mathbf{f}_{l-1}$ and $\mathbf{f}_l$ are $k[\z]$-linearly dependent. Thus $\mathbf{w}\mathbf{F}=\mathbf{0}_{1\times m}$ for some nonzero row vector $\mathbf{w}=[w_1,\ldots,w_l]\in k[\z]^{1\times l}$, where $w_l \neq 0$ and ${\rm GCD}(w_1,\ldots,w_l) = 1$. Obviously, $w_1,\ldots,w_l$ are all the $1\times 1$ reduced minors of $\mathbf{w}$.

 The next thing is to prove that $w_1,\ldots,w_l$ generate $k[\z]$. Let $\mathbf{F}_1,\ldots,\mathbf{F}_\beta \in k[\z]^{l\times (l-1)}$ be all the $l\times (l-1)$ submatrices of $\mathbf{F}$, where $\beta=\binom m{l-1}$. For each $1\leq i \leq \beta$, let $c_{i1},\ldots,c_{il}$ and $h_{i1},\ldots,h_{il}$ be all the $(l-1)\times (l-1)$ minors and all the $(l-1)\times (l-1)$ reduced minors of $\mathbf{F}_i$ respectively, then $c_{ij} = d_{l-1}(\mathbf{F}_i)\cdot h_{ij}$, where $1\leq j \leq l$. Let $\mathbf{w}=[\mathbf{w}_1,w_l]$, where $\mathbf{w}_1=[w_1,\ldots,w_{l-1}]\in k[\z]^{1\times (l-1)}$. Let $\mathbf{F}_i=[\mathbf{F}_{i1}^{{\rm T}},-\mathbf{F}_{i2}^{{\rm T}}]^{{\rm T}}$, where $\mathbf{F}_{i1}\in k[\z]^{(l-1)\times (l-1)}$ and $\mathbf{F}_{i2}\in k[\z]^{1\times (l-1)}$. If $\mathbf{F}_i$ is not of full column rank, then $c_{ij} =  0$ and $h_{ij} = 0$, $j=1,\ldots,l$. If $\mathbf{F}_i$ is of full column rank, then it follows from $\mathbf{w}\mathbf{F}=\mathbf{0}_{1\times m}$ that
 \begin{equation}\label{lemma_new_equation_1}
  \begin{bmatrix}\mathbf{w}_1,w_l\end{bmatrix}\begin{bmatrix}\mathbf{F}_{i1} \\ -\mathbf{F}_{i2} \end{bmatrix}=\mathbf{0}_{1 \times (l-1)}.
 \end{equation}
 Since $w_l \neq 0$, ${\rm det}(\mathbf{F}_{i1})\neq 0$ by Lemma \ref{RM_relation-1}. From equation (\ref{lemma_new_equation_1}) we have
 \begin{equation}\label{lemma_new_equation_2}
   w_l^{-1}\mathbf{w}_1=\mathbf{F}_{i2}\mathbf{F}_{i1}^{-1}.
 \end{equation}
 According to Lemma \ref{RM_relation-2}, all the $1\times 1$ reduced minors of $\mathbf{w}$ are equal to all the $(l-1)\times (l-1)$ reduced minors of $\mathbf{F}_i$ without considering the sign, i.e., $w_j= h_{ij}$ for $j=1,\ldots,l$. Therefore, all the $(l-1)\times (l-1)$ minors of $\mathbf{F}$ are as follows:
 $$d_{l-1}(\mathbf{F}_1)\cdot w_1,\ldots,d_{l-1}(\mathbf{F}_1)\cdot w_l,
 \cdots,d_{l-1}(\mathbf{F}_\beta)\cdot w_1,\ldots,d_{l-1}(\mathbf{F}_\beta)\cdot w_l.$$
 Let $\bar{d}\in k[\z]$ be the GCD of $d_{l-1}(\mathbf{F}_1),\ldots,d_{l-1}(\mathbf{F}_\beta)$, then there exists $\bar{d}_i \in k[\z]$ such that $d_{l-1}(\mathbf{F}_i) = \bar{d}\cdot \bar{d}_i$, where $i=1,\ldots,\beta$. In the following we prove that the polynomials
  \[\begin{matrix}
   \bar{d}_1 w_1, & \bar{d}_1 w_2, & \cdots & \bar{d}_1 w_l, \\
   \bar{d}_2 w_1, & \bar{d}_2 w_2, & \cdots & \bar{d}_2 w_l, \\
   \vdots    &  \vdots  &    \ddots    &     \vdots       \\
   \bar{d}_\beta w_1, & \bar{d}_\beta w_2, &\cdots & \bar{d}_\beta w_l,
   \end{matrix}
  \]
  are all the $(l-1)\times (l-1)$ reduced minors of $\mathbf{F}$. It follows from ${\rm GCD}(w_1,\ldots,w_l) =1$ and ${\rm GCD}(\bar{d}_1,\cdots,\bar{d}_\beta)=1$ that
  \begin{equation*}
   \begin{split}
        & {\rm GCD}(\bar{d}_1w_1,\ldots,\bar{d}_1w_l,\cdots,
            \bar{d}_\beta w_1,\ldots,\bar{d}_\beta w_l) \\
      = & {\rm GCD}({\rm GCD}(\bar{d}_1w_1,\ldots,\bar{d}_1w_l),\cdots,
            {\rm GCD}(\bar{d}_\beta w_1,\ldots,\bar{d}_\beta w_l)) \\
      = & {\rm GCD}(\bar{d}_1,\cdots,\bar{d}_\beta) \\
      = & 1.
   \end{split}
  \end{equation*}
 Therefore, $\bar{d}_1w_1,\ldots,\bar{d}_1w_l,\cdots, \bar{d}_\beta w_1,\ldots,\bar{d}_\beta w_l$ are all the $(l-1)\times (l-1)$ reduced minors of $\mathbf{F}$, i.e., they are equal to $h_1,\ldots,h_\gamma$. Since $\langle h_1,\ldots,h_\gamma \rangle = k[\z]$, $w_1,\ldots,w_l$ generate $k[\z]$. 
\end{proof}

 Combining Lemma \ref{lemma_zero} and Lemma \ref{lemma_new}, we can answer Question \ref{question-1}.

\begin{theorem} \label{theorem_new_1}
 Let $\mathbf{F}\in \mathcal{S}$. If $\langle d,h_1,\ldots,h_\gamma \rangle = k[\z]$, then $\mathbf{F}$ admits a matrix factorization w.r.t. $d$.
\end{theorem}
\begin{proof}
 We divide our proof into three steps.

 First, let $\hat{\mathbf{F}} = \mathbf{F}(f(\z_2),\z_2) \in k[\z_2]^{l\times m}$, and we prove that ${\rm rank}(\hat{\mathbf{F}}) = l-1$. Let $\hat{a}_1,\ldots,\hat{a}_\eta \in k[\z_2]$ and $\hat{c}_1,\ldots,\hat{c}_\gamma \in k[\z_2]$ be all the $l\times l$ minors and all the $(l-1)\times (l-1)$ minors of $\hat{\mathbf{F}}$, respectively. Then, $\hat{a}_i = a_i(f(\z_2),\z_2)$ and $\hat{c}_j = c_j(f(\z_2),\z_2)$, where $1\leq i \leq \eta$ and $1\leq j \leq \gamma$. Since $\mathbf{F}\in \mathcal{S}$, we have $d \mid d_l(\mathbf{F})$ and ${\rm GCD}(d,d_{l-1}(\mathbf{F}))=1$. $d \mid d_l(\mathbf{F})$ implies that $\hat{a}_i = a_i(f(\z_2),\z_2) = 0~~(i =1,\ldots,\eta)$ and ${\rm rank}(\hat{\mathbf{F}}) \leq l-1$. If ${\rm rank}(\hat{\mathbf{F}}) < l-1$, then $c_j(f(\z_2),\z_2) = \hat{c}_j =0~~(j =1,\ldots,\gamma)$. It follows from Lemma \ref{lemma_zero} that $d$ is a common divisor of $c_1,\ldots,c_\gamma$, then $d \mid d_{l-1}(\mathbf{F})$, which contradicts ${\rm GCD}(d,d_{l-1}(\mathbf{F}))=1$. Therefore, ${\rm rank}(\hat{\mathbf{F}}) = l-1$.

 Second, we prove that all the $(l-1)\times (l-1)$ reduced minors of $\hat{\mathbf{F}}$ generate $k[\z_2]$. Let $\bar{h}\in k[\z_2]$ be the GCD of $h_1(f(\z_2),\z_2),\ldots,h_\gamma(f(\z_2),\z_2)$, then for each $1\leq j \leq \gamma$ there exits $\hat{h}_j\in k[\z_2]$ such that $h_j(f(\z_2),\z_2)=\bar{h}\cdot \hat{h}_j$, and ${\rm GCD}(\hat{h}_1,\ldots, \hat{h}_\gamma)=1$. Let $g = d_{l-1}(\mathbf{F})$, then it follows from $\hat{c}_j = g(f(\z_2),\z_2)\cdot h_j(f(\z_2),\z_2)$ that $d_{l-1}(\hat{\mathbf{F}}) = g(f(\z_2),\z_2)\cdot \bar{h}$, and $\hat{h} _1,\ldots,\hat{h}_\gamma$ are all the $(l-1)\times (l-1)$ reduced minors of $\hat{\mathbf{F}}$. Assume that $\langle \hat{h} _1,\ldots,\hat{h}_\gamma \rangle \neq k[\z_2]$, then there exists a point $(\alpha_2,\ldots,\alpha_n)\in k^{1\times (n-1)}$ such that $\hat{h}_j(\alpha_2,\ldots,\alpha_n)=0$, where $j =1, \ldots, \gamma$. Let $\alpha_1 = f(\alpha_2,\ldots,\alpha_n)$, then for each $j$ we have $h_j(\alpha_1,\alpha_2,\ldots,\alpha_n) =\bar{h}(\alpha_2,\ldots,\alpha_n)
 \cdot \hat{h}_j(\alpha_2,\ldots,\alpha_n) = 0$. This implies that $(\alpha_1,\alpha_2,\ldots,\alpha_n)\in k^{1\times n}$ is a common zero of $d,h_1,\ldots,h_\gamma$, which contradicts the fact that $\langle d,h_1,\ldots,h_\gamma \rangle = k[\z]$.

 Finally, we remark that $\mathbf{F}$ has a matrix factorization w.r.t. $d$. Using Lemma \ref{lemma_new}, we get $\mathbf{w}\hat{\mathbf{F}}= \mathbf{0}_{1\times m}$, where $\mathbf{w}\in k[\z_2]^{1 \times l}$ is a ZLP vector. Meanwhile, according to Lemma \ref{QS-2}, a unimodular matrix $\mathbf{U} \in k[\z_2]^{l \times l}$ can be constructed such that $\mathbf{w}$ is its first row. Let $\mathbf{F}_0 = \mathbf{U}\mathbf{F}$, then the first row of $\mathbf{F}_0(f(\z_2),\z_2)= \mathbf{U}\hat{\mathbf{F}}$ is a zero vector. By Lemma \ref{lemma_zero}, $d$ is a common divisor of the polynomials in the first row of $\mathbf{F}_0$, thus
 \[\mathbf{F}_0=\mathbf{U}\mathbf{F} = \mathbf{D} \mathbf{F}_1=\begin{bmatrix}
     d  &   &   &  &       &  &  &   \\
          &   & 1 &  &       &  &  &   \\
          &   &   &  & \ddots&  &  &   \\
          &   &   &  &       &  &  & 1
   \end{bmatrix}
   \begin{bmatrix}
     \bar{f}_{11}  & &  & & \cdots & &  & & \bar{f}_{1m}  \\
     f_{21}  & &  & & \cdots & &  & & f_{2m}  \\
     \vdots  & &  & & \ddots & &  & & \vdots   \\
      f_{l1} & &  & & \cdots & &  & & f_{lm}
   \end{bmatrix}.
 \]
 Consequently, we can now derive the matrix factorization of $\mathbf{F}$ w.r.t. $d$:
 $$\mathbf{F}= \mathbf{G}_1\mathbf{F}_1,$$
 where $\mathbf{G}_1= \mathbf{U}^{-1}\mathbf{D} \in k[\z]^{l\times l}$, $\mathbf{F}_1\in k[\z]^{l\times m}$ and ${\rm det}(\mathbf{G}_1)=d$. 
\end{proof}

\begin{remark}\label{theorem_new_1_remark}
 In Theorem \ref{theorem_new_1}, we have that ${\rm rank}(\hat{\mathbf{F}}) = l-1$ and $\langle \hat{h} _1,\ldots,\hat{h}_\gamma \rangle = k[\z_2]$. Hence, Theorem \ref{theorem_new_1} is a generalization of Lemma \ref{theorem_Liu}.
\end{remark}

 According to Theorem \ref{theorem_new_1}, we construct a set of multivariate polynomial matrices:
 $$\mathcal{S}_3 := \{\mathbf{F}\in \mathcal{S} : \langle d,h_1,\ldots,h_\gamma \rangle = k[\z]\}.$$
 Then, $\mathcal{S}_2 \subset \mathcal{S}_3 \subset \mathcal{S}$ and $\mathbf{F}\in \mathcal{S}_3$ admits a matrix factorization w.r.t. $d$. Example \ref{example-1} in Section \ref{sec2-2} shows that $S_2 \subsetneqq S_3$.

 \vspace{4mm}
 Let $\mathbf{F}\in k[\z]^{l\times m}$, and $d_0 = \prod_{t=1}^{s}(z_1 - f_t(\z_2))$ be a divisor of $d_l(\mathbf{F})$, where $f_1(\z_2),\ldots,f_s(\z_2)\in k[\z_2]$. \cite{Liu2011On} proved that if $\langle d_0,c_1,\ldots,c_\gamma \rangle = k[\z]$, then $\mathbf{F}$ admits a matrix factorization w.r.t. $d_0$. It would be interesting to know whether Theorem \ref{theorem_new_1} can be generalized to the case with $t>1$. Without loss of generality, we consider the case of $t=2$.

\begin{theorem}\label{theorem_new_2}
 Let $\mathbf{F}\in k[\z]^{l\times m}$ and $d_0 = (z_1 - f_1(\z_2))(z_1 - f_2(\z_2))$ be a divisor of $d_l(\mathbf{F})$. If ${\rm GCD}(d_0,d_{l-1}(\mathbf{F}))=1$ and $\langle d_0,h_1,\ldots,h_\gamma \rangle = k[\z]$, then $\mathbf{F}$ admits a matrix factorization w.r.t. $d_0$.
\end{theorem}
\begin{proof}
 Let $d_1 = z_1 - f_1(\z_2)$ and $d_2 = z_1 - f_2(\z_2)$. Obviously, ${\rm GCD}(d_1,d_{l-1}(\mathbf{F}))=1$ and $\langle d_1,h_1,\ldots,h_\gamma \rangle = k[\z]$. By Theorem \ref{theorem_new_1}, there exist $\mathbf{G}_1\in k[\z]^{l\times l}$ and $\mathbf{F}_1\in k[\z]^{l\times m}$ such that
 $$\mathbf{F} = \mathbf{G}_1\mathbf{F}_1,$$
 where $\mathbf{G}_1= \mathbf{U}_1^{-1}\mathbf{D}_1$, ${\rm det}(\mathbf{G}_1)=d_1$, $\mathbf{U}_1\in k[\z_2]^{l\times l}$ is a unimodular matrix and $\mathbf{D}_1={\rm diag}(d_1,1,\ldots,1)$. According the Equation (\ref{binet-cauchy-equation-2}) in Lemma \ref{binet-cauchy}, $d_2 = z_1 - f_2(\z_2)$ is a divisor of $d_l(\mathbf{F}_1)$. Next we prove that $\mathbf{F}_1$ admits a matrix factorization w.r.t. $d_2$.

 We first prove that ${\rm GCD}(d_2, d_{l-1}(\mathbf{F}_1))=1$. Otherwise, it follows from $d_2$ is an irreducible polynomial that ${\rm GCD}(d_2, d_{l-1}(\mathbf{F}_1))=d_2$. Then $d_{l-1}(\mathbf{F}_1) \mid d_{l-1}(\mathbf{F})$ implies that $d_2 \mid d_{l-1}(\mathbf{F})$, which contradicts the condition ${\rm GCD}(d,d_{l-1}(\mathbf{F}))=1$. Second, we prove that $d_2$ and all the $(l-1)\times (l-1)$ reduced minors of $\mathbf{F}_1$ generate the unit ideal $k[\z]$.

 Let $\mathbf{F}_{i1}\in k[\z]^{(l-1)\times m}$ be a submatrix obtained by removing the $i$-th row of $\mathbf{F}_1$, and $\bar{c}_{i1},\ldots,\bar{c}_{i \beta}$ be all the $(l-1)\times (l-1)$ minors of $\mathbf{F}_{i1}$, where $i=1,\ldots,l$. Then, $\bar{c}_{11},\ldots,\bar{c}_{1 \beta},\ldots,\bar{c}_{l1}, \ldots,\bar{c}_{l\beta}$ are all the $(l-1)\times (l-1)$ minors of $\mathbf{F}_1$. Extracting $d_{l-1}(\mathbf{F}_1)$ from $\bar{c}_{ij}$ yields $\bar{c}_ {ij} = d_{l-1}(\mathbf{F}_1)\cdot \bar{h}_{ij}$, then $\bar{h}_{11},\ldots,\bar{h}_{1 \beta},\ldots,\bar{h}_{l1}, \ldots,\bar{h}_{l\beta}$ are all the $(l-1)\times (l-1)$ reduced minors of $\mathbf{F}_1$. Hence, we only need to prove that $\langle d_2,\bar{h}_{11},\ldots,\bar{h}_{l \beta} \rangle = k[\z]$.

 Since $\mathbf{D}_1={\rm diag}(d_1,1,\ldots,1)$, all the $(l-1)\times (l-1)$ minors of $\mathbf{D}_1 \mathbf{F}_1$ are
 $$\bar{c}_{11},\ldots,\bar{c}_{1 \beta},d_1\bar{c}_{21}, \ldots,d_1\bar{c}_{2\beta},\ldots,d_1\bar{c}_{l1}, \ldots,d_1\bar{c}_{l\beta}.$$
 Obviously, there is at least one integer $j\in \{1,\ldots,\beta\}$ such that $d_1 \nmid \bar{c}_{1j}$. Otherwise, $d_1 \mid d_{l-1}(\mathbf{D}_1 \mathbf{F}_1)$. It follows form $\mathbf{F}=\mathbf{U}_1^{-1}\mathbf{D}_1 \mathbf{F}_1$ and the Equation (\ref{binet-cauchy-equation-1}) in Lemma \ref{binet-cauchy} that $d_{l-1}(\mathbf{D}_1 \mathbf{F}_1) \mid d_{l-1}(\mathbf{F})$. So $d_1 \mid d_{l-1}(\mathbf{F})$, which leads to a contradiction. Since $d_1 = z_1 - f_1(\z_2)$ is an irreducible polynomial, we have
 \begin{equation*}
   \begin{split}
        & {\rm GCD}(\bar{c}_{11},\ldots,\bar{c}_{1 \beta},d_1\bar{c}_{21}, \ldots,d_1\bar{c}_{2\beta},\ldots,d_1\bar{c}_{l1}, \ldots,d_1\bar{c}_{l\beta}) \\
    =   & {\rm GCD}(\bar{c}_{11},\ldots,\bar{c}_{1 \beta},\bar{c}_{21}, \ldots,\bar{c}_{2\beta},\ldots,\bar{c}_{l1}, \ldots,\bar{c}_{l\beta}),
   \end{split}
 \end{equation*}
 Therefore, $d_{l-1}(\mathbf{D}_1 \mathbf{F}_1) = d_ {l-1}(\mathbf{F}_1)$. It follows from $\mathbf{U}_1\mathbf{F}=\mathbf{D}_1 \mathbf{F}_1$ that $d_ {l-1}(\mathbf{F}) \mid d_{l-1}(\mathbf{D}_1 \mathbf{F}_1)$ and $d_{l-1}(\mathbf{F}) = d_ {l-1}(\mathbf{F}_1)$. The Equation (\ref{binet-cauchy-equation-1}) in Lemma \ref{binet-cauchy} implies that each $c_i$ is a $k[\z]$-linear combination of $\bar{c}_{11},\ldots,\bar{c}_{l\beta}$, where $i=1,\ldots,\gamma$. Since $d_{l-1}(\mathbf{F}) = d_ {l-1}(\mathbf{F}_1)$, we can obtain that each $h_i$ $(1\leq i \leq \gamma)$ is a $k[\z]$-linear combination of $\bar{h}_{11},\ldots,\bar{h}_{l\beta}$. By $\langle d_0,h_1,\ldots,h_\gamma \rangle = k[\z]$, $\langle d_2,h_1,\ldots,h_\gamma \rangle = k[\z]$. If $\langle d_2,\bar{h}_{11},\ldots,\bar{h}_{l \beta} \rangle \neq k[\z]$, then there exits a point $(\alpha_1,\ldots,\alpha_n)\in k^{1\times n}$ such that $\bar{h}_{ij}(\alpha_1,\ldots,\alpha_n) = 0$ for each $i$ and $j$, where $\alpha_1 = f_2(\alpha_2,\ldots,\alpha_n)$. This implies that $(\alpha_1,\ldots,\alpha_n)$ is a common zero of $d_2,h_1,\ldots,h_\gamma$, which leads to a contradiction.

 According to Theorem \ref{theorem_new_1} again, there exits $\mathbf{G}_2\in k[\z]^{l\times l}$ and $\mathbf{F}_2\in k[\z]^{l\times m}$ such that $\mathbf{F}_1=\mathbf{G}_2\mathbf{F}_2$, where $\mathbf{G}_2= \mathbf{U}_2^{-1}\mathbf{D}_2$, ${\rm det}(\mathbf{G}_2)=d_2$, $\mathbf{U}_2\in k[\z_2]^{l\times l}$ is an unimodular matrix and $\mathbf{D}_2={\rm diag}(d_2,1,\ldots,1)$.

 Finally, we can get a matrix factorization of $\mathbf{F}$ w.r.t. $d_0$:
 $$\mathbf{F} = \mathbf{G}_0\mathbf{F}_2,$$
 where $\mathbf{G}_0 = \mathbf{G}_1\mathbf{G}_2 \in k[\z]^{l\times l}$, and ${\rm det}(\mathbf{G}_0)=d_0 = (z_1 - f_1(\z_2))(z_1 - f_2(\z_2))$. 
\end{proof}

\begin{remark}\label{theorem_new_2_remark}
 In the above theorem, we can factorize $\mathbf{F}_1$ w.r.t. $d_2$ without checking whether ${\rm GCD}(d_2,d_{l-1}(\mathbf{F}_1))=1$ and the ideal generated by $d_2$ and all the $(l-1)\times (l-1)$ reduced minors of $\mathbf{F}_1$ is $k[\z]$, which can help us improve the computational efficiency of matrix factorizations.
\end{remark}

 It is worth noting that if $f_1(\z_2) = f_2(\z_2)$ in Theorem \ref{theorem_new_2},  we have the following corollary.

\begin{corollary}  \label{new_cor_1}
 Let $\mathbf{F}\in k[\z]^{l\times m}$ and $d_0 = (z_1 - f_1(\z_2))^r$ be a divisor of $d_l(\mathbf{F})$. If ${\rm GCD}(d_0,d_{l-1}(\mathbf{F}))=1$ and $\langle d_0,h_1,\ldots,h_\gamma \rangle = k[\z]$, then $\mathbf{F}$ admits a matrix factorization w.r.t. $d_0$.
\end{corollary}

 Further, if $f_1(\z_2) \neq f_2(\z_2)$ in Theorem \ref{theorem_new_2}, we have another corollary.

\begin{corollary} \label{new_cor_2}
 Let $\mathbf{F}\in k[\z]^{l\times m}$ and $d_0 = \prod_{t=1}^{s}(z_1 - f_t(\z_2))^{q_t}$ be a divisor of $d_l(\mathbf{F})$. If ${\rm GCD}(d_0,d_{l-1}(\mathbf{F}))$ $=1$ and $\langle d_0,h_1,\ldots,h_\gamma \rangle = k[\z]$, then $\mathbf{F}$ admits a matrix factorization w.r.t. $d_0$.
\end{corollary}

 Let $f^{(i)}(\z)$ be a polynomial in $k[z_1,\ldots,z_{i-1},z_{i+1},
 \ldots,z_n]$, where $1 \leq i \leq n$. Similarly, we can get the following corollaries.

\begin{corollary}  \label{new_cor_3}
 Let $\mathbf{F}\in k[\z]^{l\times m}$ and $d_0 = (z_i - f^{(i)}(\z))^r$ be a divisor of $d_l(\mathbf{F})$. If ${\rm GCD}(d_0,d_{l-1}(\mathbf{F}))=1$ and $\langle d_0,h_1,\ldots,h_\gamma \rangle = k[\z]$, then $\mathbf{F}$ admits a matrix factorization w.r.t. $d_0$.
\end{corollary}

\begin{corollary} \label{new_cor_4}
 Let $\mathbf{F}\in k[\z]^{l\times m}$ and $d_0 = \prod_{i=1}^{n} \prod_{t=1}^{s_i}(z_i - f^{(i)}_t(\z))^{q_{it}}$ be a divisor of $d_l(\mathbf{F})$. If ${\rm GCD}(d_0,d_{l-1}(\mathbf{F}))=1$ and $\langle d_0,h_1,\ldots,h_\gamma \rangle = k[\z]$, then $\mathbf{F}$ admits a matrix factorization w.r.t. $d_0$.
\end{corollary}

\section{Algorithm and Example} \label{sec4}

 According to Theorem \ref{theorem_new_1}, we get the following algorithm for computing a matrix factorization of $\mathbf{F}\in \mathcal{S}_3$ w.r.t. $d$.

%  \vskip -14 pt

 \begin{algorithm}[!htb]
 \DontPrintSemicolon
 \SetAlgoSkip{}
% \SetAlgoNoLine
 \LinesNumbered
 \SetKwInOut{Input}{Input}
 \SetKwInOut{Output}{Output}

 \Input{$\mathbf{F}\in \mathcal{S}_3$.}

 \Output{a matrix factorization of $\mathbf{F}$ w.r.t. $d$.}

 \Begin{

   compute a ZLP vector $\mathbf{w}\in k[\z_2]^{1 \times l}$ such that $\mathbf{w}\mathbf{F}(f(\z_2),\z_2)=\mathbf{0}_{1\times m}$;

   construct a unimodular matrix $\mathbf{U}\in k[\z_2]^{l \times l} $ such that $\mathbf{w}$ is its first row;

   compute $\mathbf{F}_1\in k[\z]^{l\times m}$ such that $\mathbf{U}\mathbf{F}=\mathbf{D} \mathbf{F}_1$, where $\mathbf{D}={\rm diag}(d,1,\ldots,1)$;

   {\bf return} $\mathbf{F} = \mathbf{G}_1\mathbf{F}_1$, where $\mathbf{G}_1= \mathbf{U}^{-1}\mathbf{D}$ and ${\rm det}(\mathbf{G}_1)=d$.
 }
 \caption{Matrix Factorization Algorithm}
 \label{MF_Algorithm}
 \end{algorithm}

 % \vskip -16 pt

 In the following, we show how to compute $\mathbf{w}$ and $\mathbf{U}$ in Algorithm \ref{MF_Algorithm}. Let $\hat{\mathbf{F}} = \mathbf{F}(f(\z_2),\z_2) \in k[\z_2]^{l\times m}$ and ${\rm Syz}_l(\hat{\mathbf{F}})$ be the left syzygy module of $\hat{\mathbf{F}}$, i.e., ${\rm Syz}_l(\hat{\mathbf{F}})=\{\mathbf{p}\in k[\z_2]^{1\times l}\mid \mathbf{p}\hat{\mathbf{F}}=\mathbf{0}_{1\times m}\}$. Since ${\rm rank}(\hat{\mathbf{F}}) = l-1$, we have ${\rm rank}({\rm Syz}_l(\hat{\mathbf{F}})) = 1$. Then, we compute a reduced \gr basis of ${\rm Syz}_l(\hat{\mathbf{F}})$ w.r.t. a term order, and select a nonzero vector from the \gr basis. Let $\mathbf{w}_1=[w_{11},\ldots,w_{1l}]\in k[\z_2]^{1 \times l}$ be the nonzero vector, and $w\in k[\z_2]$ be the GCD of $w_{11},\ldots,w_{1l}$, then $\mathbf{w} = \frac{\mathbf{w}_1}{w}$.

 Since $\mathbf{w}$ is a ZLP vector, there exists a column vector $\mathbf{q}_1 \in k[\z_2]^{l\times 1}$ such that $\mathbf{w}\mathbf{q}_1 =1$. This calculation problem is equivalent to a lifting homomorphism problem in \citep{Decker2006Computing} (see Problem 4.1, page 129), and the command ``lift" of the computer algebra system {\em Singular} in \citep{DGPS2016} can help us compute $\mathbf{q}_1$. Let ${\rm Syz}_r(\mathbf{w})=\{\mathbf{q}\in k[\z_2]^{l\times 1}\mid \mathbf{w}\mathbf{q}=0\}$, then ${\rm Syz}_r(\mathbf{w})$ is a free module with ${\rm rank}({\rm Syz}_r(\mathbf{w}))=(l-1)$ by Lemma \ref{QS-1}. Let $\mathbf{q}_2,\ldots, \mathbf{q}_l\in k[\z_2]^{l\times 1}$ be a free basis of ${\rm Syz}_r(\mathbf{w})$, then $\mathbf{V} = [\mathbf{q}_1,\mathbf{q}_2,\ldots, \mathbf{q}_l] \in k[\z_2]^{l\times l}$ is a unimodular matrix and $\mathbf{U} = \mathbf{V}^{-1}$ is one that we want by Theorem 4.4 in \citep{Lu2017}.

 Now, we use an example to illustrate the calculation process of Algorithm \ref{MF_Algorithm}. We return to Example \ref{example-1}, and let $\mathbf{F}$ be the same matrix in Example \ref{example-1}.

\begin{example}\label{example-1-1}
 Let
 \begin{small}
  \[\mathbf{F}=\begin{bmatrix}
  z_1z_2-z_1-z_2^2-z_2z_3  &  z_1z_3+z_1-z_2z_3-z_2-z_3^2-z_3  & \mathbf{F}[1,3] \\
  -z_1z_2-z_1z_3+z_2+z_3 & z_2+z_3  &  z_1z_2+z_1z_3  \\
  z_1   &       -z_1+z_2+z_3 &         -2z_1+z_2+z_3+1
  \end{bmatrix}, \]
 \end{small}
 where $\mathbf{F}[1,3] = -z_1z_2+z_1z_3+2z_1+z_2^2-z_2-z_3^2-2z_3-1$.

 As already noted in Example \ref{example-1}, $\langle d, h_1,\ldots,h_9 \rangle = k[z_1,z_2,z_3]$ implies that $\mathbf{F}\in \mathcal{S}_3$. Then, we can use Algorithm \ref{MF_Algorithm} to factorize $\mathbf{F}$ w.r.t. $d$.

\vspace{1mm}
 {\bf Step 1:} Let $\hat{\mathbf{F}} = \mathbf{F}(z_2,z_2,z_3) \in k[z_2,z_3]^{3\times 3}$, we compute a ZLP vector $\mathbf{w}\in k[z_2,z_3]^{1 \times 3}$ such that $\mathbf{w}\hat{\mathbf{F}} = \mathbf{0}_{1\times 3}$, where
 \[\hat{\mathbf{F}} =\begin{bmatrix}
   -z_2(z_3+1) & -z_3(z_3+1) & -(z_3-z_2+1)(z_3+1) \\
   (1-z_2)(z_2+z_3) & z_2+z_3 & z_2(z_2+z_3) \\
   z_2 & z_3 & z_3-z_2+1
 \end{bmatrix}. \]
 We use {\em Singular} command ``syz" to compute a reduced \gr basis of ${\rm Syz}_l(\hat{\mathbf{F}})$ w.r.t. the lexicographic order, and obtain $\mathbf{w} =[1,0,z_3+1]$.

\vspace{1mm}
 {\bf Step 2:} Construct a unimodular matrix $\mathbf{U}\in k[z_2,z_3]^{3 \times 3} $ such that $\mathbf{w}$ is its first row. According to the instruction of the construction for unimodular matrix $\mathbf{U}$ below Algorithm \ref{MF_Algorithm}, we divide it into three small steps.

  Step 2.1: Using {\em Singular} command ``lift" to compute $\mathbf{q}_1 \in k[z_2,z_3]^{3\times 1}$ such that $\mathbf{w}\mathbf{q}_1 =1$, we get $\mathbf{q}_1 = [1,0,0]^{\rm T}$.

 Step 2.2: Using QuillenSuslin package to compute a free basis of ${\rm Syz}_r(\mathbf{w})$, we have $\mathbf{q}_2=[0,1,0]^{\rm T}$ and $\mathbf{q}_3=[-(z_3+1),0,1]^{\rm T}$.

 Step 2.3: Let $\mathbf{V}=[\mathbf{q}_1,\mathbf{q}_2,\mathbf{q}_3]$, then
 \[\mathbf{U} = \mathbf{V}^{-1} =\begin{bmatrix}
   1 & 0 & z_3+1 \\
   0 & 1 & 0 \\
   0 & 0 & 1
 \end{bmatrix}. \]

\vspace{1mm}
 {\bf Step 3.} Extracting $d$ from the first row of $\mathbf{U}\mathbf{F}$, we get $\mathbf{U}\mathbf{F} =\mathbf{D} \mathbf{F}_1$, where $\mathbf{D}={\rm diag}(d,1,1)$ and
 \[\mathbf{F}_1 = \begin{bmatrix}
  z_2+z_3  &  0  & -z_2-z_3 \\
  -z_1z_2-z_1z_3+z_2+z_3 & z_2+z_3  &  z_1z_2+z_1z_3  \\
  z_1   &       -z_1+z_2+z_3 &         -2z_1+z_2+z_3+1
 \end{bmatrix}. \]
 Then, we obtain a matrix factorization of $\mathbf{F}$ w.r.t. $d$:
 \begin{small}
  \[\mathbf{F} = \mathbf{G}\mathbf{F}_1 = \begin{bmatrix}
     d & 0 & -z_3 - 1 \\
     0 & 1 & 0 \\
     0 & 0 & 1
  \end{bmatrix}
  \begin{bmatrix}
     z_2 + z_3 & 0 & -z_2 - z_3 \\
     -z_1z_2-z_1z_3+z_2+z_3 & z_2+z_3  &  z_1z_2+z_1z_3  \\
     z_1   &       -z_1+z_2+z_3 &         -2z_1+z_2+z_3+1
  \end{bmatrix},\]
 \end{small}
 where $\mathbf{G}=\mathbf{U}^{-1}\mathbf{D}$ and ${\rm det}(\mathbf{G}_1) = d = z_1-z_2$.

\end{example}

\section{Conclusions} \label{sec_conclusions}

 We have studied the problem of matrix factorizations for multivariate polynomial matrices in $\mathcal{S}$, and the results presented in this paper greatly extend those of \citep{Lin2001Factorizations,Liu2011On}. The matrix factorizations for an arbitrary multivariate polynomial matrix remains a challenging and an important open problem. Although the new results can only deal with the class of multivariate polynomial matrices discussed in $\mathcal{S}$, we hope that the new results will motivate new progress in this important research topic.

\vspace{10pt}
\noindent \textbf{Acknowledgements\ }
This research was supported by the CAS Project QYZDJ-SSW-SYS022.

%\vspace{10pt}
%
%\noindent {\bf References}

\bibliographystyle{elsarticle-harv}

\bibliography{mf} 

\end{document}